\documentclass[letterpaper, 10pt, conference]{ieeeconf}  
\IEEEoverridecommandlockouts                              
\usepackage[applemac]{inputenc}
\usepackage{graphics} 
\usepackage{epsfig} 
\usepackage{epstopdf}
\usepackage{amsmath} 
\usepackage{amssymb}  
\usepackage{epstopdf}
\usepackage{graphicx,xcolor}
\usepackage{caption}
\usepackage{subcaption}
\usepackage{courier}
\usepackage{mathrsfs}

\newtheorem{thm}{Theorem}

\newtheorem{cor}{Corollary}

\newtheorem{prop}{Proposition}
\newtheorem{problem}{Problem}
\newtheorem{example}{Example}
\newcommand{\nt}{{\mathbb N}}

\title{\LARGE \bf
Verification of Uncertain POMDPs Using Barrier Certificates
}

\author{Mohamadreza Ahmadi, Murat Cubuktepe, Nils Jansen, and Ufuk Topcu
\thanks{M. Ahmadi, M. Cubuktepe, and U. Topcu are with the Department of Aerospace Engineering and Engineering Mechanics, and the Institute for Computational Engineering and Sciences (ICES), University of Texas, Austin, 201 E 24th St, Austin, TX 78712. N. Jansen is with the Radboud University Nijmegen, The Netherlands. e-mails: (\{mrahmadi, mcubuktepe, utopcu\}@utexas.edu), n.jansen@science.ru.nl. 
}
}

\begin{document}

\maketitle
\thispagestyle{empty}
\pagestyle{empty}

\begin{abstract}

We consider a class of partially observable Markov decision processes (POMDPs) with uncertain transition and/or observation probabilities.
The uncertainty takes the form of probability intervals.
Such uncertain POMDPs can be used, for example, to model autonomous agents with sensors with limited accuracy, or undergoing a sudden component failure, or structural damage~\cite{MPSS10}.  
Given an uncertain POMDP representation of the autonomous agent, our goal is to propose a method for checking whether the system will satisfy an optimal performance, while not violating a safety requirement (e.g. fuel level,  velocity, and etc.).
To this end, we cast the POMDP problem into a switched system scenario. 
We then take advantage of this switched system characterization and propose a method based on barrier certificates for optimality and/or safety verification.  
We then show that the verification task can be carried out computationally by sum-of-squares programming. 
We illustrate the efficacy of our method by applying it to a Mars rover exploration example.
\end{abstract}

\section{INTRODUCTION}
A popular formal model for planning subject to stochastic behavior are Markov decision processes (MDPs)~\cite{Put94}, where an agent chooses to perform an action under full knowledge of the environment it is operating in.
The outcome of the action is a probability distribution over the system states.
Many applications, however, allow only \emph{partial observability} of the current system state~\cite{kaelbling1998planning,thrun2005probabilistic,WongpiromsarnF12}.
Partially observable Markov decision processes (POMDPs) extend MDPs to account for such partial information~\cite{Russell-AI-Modern}.
Upon certain \emph{observations}, the agent infers the likelihood of the system being in a certain state, called the belief state. 
The belief state together with an update function form a (typically uncountably infinite) MDP, referred to as the \emph{belief MDP}~\cite{ShaniPK13,MadaniHC99,braziunas2003pomdp,szer2005optimal,NPZ17}.

Most formulations assume the the transition probability function and the observation function for MDPs and POMDPs are explicitly given.
Unforeseeable events such as (unpredictable) structural damage to a system~\cite{796wwe3643} or an imprecise sensor model~\cite{bagnell2001solving}, however, necessitate a more robust formalism. 
So-called uncertain MDPs incorporate \emph{uncertainty-sets} of probabilities, for instance, for models that are empirically determined.
Similar extensions exist for uncertain POMDPs~\cite{burns2007sampling,itoh2007,bry2011rapidly}.

Here, we aim to address current challenges for the area of artificial intelligence, referred to as \emph{robust decision-making} and \emph{safe exploration}~\cite{russel_priorities,amodei2016concrete,stoica2017berkeley,freedman2016safety}.
Concretely, the problem is to provide a \emph{policy} for an (autonomous) agent that ensures certain desired behavior by \emph{robustly} accounting for any uncertainty and partial observability that may occur in the system~\cite{howard1960dynamic}.
The policy should be \emph{optimal} with respect to some performance measure and additionally ensure \emph{safe} navigation through the environment. 

However, already for mere POMDPs (without uncertainties in the probabilities), such policies are computed by assessing the entire belief MDP, rendering the problem undecidable~\cite{ChatterjeeCT16}. Several promising approximate point-based methods via finite abstraction of the belief space are proposed in the literature~\cite{Hauskrecht2000,Spaan2005,6284837}. Nonetheless, these techniques do not provide a guarantee for safety or optimality. That is, it is not clear whether the probability of satisfying the safety/optimality requirement is an upper-bound or a lower-bound for a given POMDP. Establishing guaranteed performance is of fundamental importance in safety-critical applications, e.g. aircraft collision avoidance~\cite{Kochenderfer2013} and Mars rovers~\cite{smith2004heuristic}.  

In this paper, we borrow a notion from control theory to provide guarantees for optimality and safety of uncertain POMDPs, without the need for finite abstraction. We first demonstrate that POMDP analysis problems can be represented as analyzing the solutions of a special discrete-time switched systems~\cite{liberzon2003switching}. In particular, POMDPs with uncertain transition and/or observation probabilities belonging to intervals can be characterized as a class of  switched systems with parametric uncertainty. Based on this switched system representation, we verify the safety and/or optimality requirements of a given POMDP using barrier certificates~\cite{4287147} (see our preliminary results with application to privacy verification of POMDPs~\cite{ahmadi2018privacy}). We show that if there exist a barrier certificate satisfying a set of inequalities along the belief update equation of the POMDP, the safety/optimality property is  guaranteed to hold. These conditions can be computationally implemented as a set of sum-of-squares programs (see Appendix~\ref{app:SOS} for a brief introduction). We elucidate the proposed method by applying it to  an uncertain POMDP model for a Mars rover with uncertain sensor accuracy sampling rocks on the Mars terrain. 

The rest of the paper is organized as follows. In the subsequent section, we define the notations used in this paper and review some preliminary definitions. In Section~\ref{sec:probform}, we describe the class of uncertainties and the properties that we are interested in. In Section~\ref{sec:main}, we propose a switched system representation for POMDPs, and present conditions based on barrier certificates for checking optimality and/or safety.  In Section~\ref{sec:example}, we apply the proposed method to verify the performance of an uncertain POMDP model of a Mars rover sampling rocks. Finally, in Section~\ref{sec:conclusions}, we conclude the paper and give directions for future research.



\section{Preliminaries}

\textbf{Notation:}
The notations employed in this paper are relatively straightforward. $\mathbb{R}_{\ge 0}$ denotes the set $[0,\infty)$.  $\mathbb{Z}$ denotes the set of integers and $\mathbb{Z}_{\ge c}$ for $c \in \mathbb{Z}$ implies the set $\{c,c+1,c+2,\ldots\}$. $\mathcal{R}[x]$ accounts for the set of polynomial functions with real coefficients in $x \in \mathbb{R}^n$, $p: \mathbb{R}^n \to \mathbb{R}$ and $\Sigma \subset\mathcal{R}$ is the subset of polynomials with a sum-of-squares decomposition; i.e, $p \in  \Sigma[x]$ if and only if there are $p_i \in \mathcal{R}[x],~i \in \{1, \ldots ,k\}$ such that $p = p_i^2 + \cdots +p_k^2$.  For a finite set $A$, $|A|$ denotes the number of elements in $A$.


\subsection{Partially Observable Markov Decision Processes}
Markov decision processes (MDPs)~\cite{Put94} are  decision-making modeling framework, in which the actions have stochastic outcomes. An MDP $\mathcal{M}=(Q,p_0,A,T)$ has the following components:
	\begin{itemize}
		\item $Q$ is a finite set of states with indices $\{1,2,\ldots,n\}$.
		\item $p_0:Q\rightarrow[0,1]$ defines the distribution of the initial states, i.e., $p_0(q)$ denotes the probability of starting at $q\in Q$.
		\item $A$ is a finite set of actions.
		\item $T:Q\times A\times Q\rightarrow [0,1]$ is the probabilistic transition function, where 
		\begin{multline}
		T(q,a,q'):=P(q_t=q'|q_{t-1}=q,a_{t-1}=a),~\\
		\forall t\in\mathbb{Z}_{\ge 1}, q,q'\in Q, a\in A. \nonumber
		\end{multline}		 
	\end{itemize}

POMDPs provide a more general mathematical framework to consider not only the stochastic outcomes of actions, but also the imperfect state observations~\cite{Sondik78}. Formally,  a POMDP $\mathcal{P}=(Q,p_0,A,T,Z,O)$ is defined with the following components:

	\begin{itemize}
		\item $Q,p_0,A,T$ are the same as the definition of an MDP.
		\item $Z$ is the set of all possible observations representing outputs of a discrete sensor. Usually $z\in Z$ is an incomplete projection of the world state $q$, contaminated by sensor noise.
		\item $O:Q\times A \times Z\rightarrow [0,1]$ is the observation probability transition function (sensor model), where
		\begin{multline}
		O(q,a,z):=P(z_t=z|q_{t}=q,a_{t-1}=a),~\\
	    \forall t\in\mathbb{Z}_{\ge 1}, q\in Q, a\in A, z\in Z. \nonumber
		\end{multline}			 
	\end{itemize}
	
	Since the states are not directly accessible in a POMDP, decision making requires the history of observations. Therefore, we need to define the notion of a \emph{belief} or the posterior as sufficient statistics for the history~\cite{astrom}. Given a POMDP, the belief at $t=0$ is defined as $b_0(q)=p_0(q)$ and $b_t(q)$ denotes the probability of system being in state $q$ at time $t$. At time $t+1$, when action $a\in A$ is observed, the belief update can be obtained by a Bayesian filter as
\begin{align} \label{equation:belief update}
b_t(q')
&=P(q'|z_t,a_{t-1},b_{t-1}) \nonumber \\
&= \frac{P(z_t|q',a_{t-1},b_{t-1})P(q'|a_{t-1},b_{t-1})}{P(z_t|a_{t-1},b_{t-1})}\nonumber \\
&= \frac{P(z_t|q',a_{t-1},b_{t-1})}{P(z_t|a_{t-1},b_{t-1})} \nonumber \\
&    \times \sum_{q\in Q}P(q'|a_{t-1},b_{t-1},q)P(q|a_{t-1},b_{t-1}) \nonumber \\
&=\frac{O(q',a_{t-1},z_{t})\sum_{q\in Q}T(q,a_{t-1},q')b_{t-1}(q)}{\sum_{q'\in Q}O(q',a_{t-1},z_{t})\sum_{q\in Q}T(q,a_{t-1},q')b_{t-1}(q)},
\end{align}
where the beliefs belong to the belief unit simplex
$$
\mathcal{B} = \left\{ b \in [0,1]^{|Q|} \mid \sum_q b_t(q)=1, \forall t  \right\}.
$$

A policy in a POMDP setting is then a mapping $\pi:\mathcal{B} \to A$, i.e., a mapping from the continuous beliefs space into the discrete and finite action space.

\section{Problem Formulation}\label{sec:probform}

We represent the uncertainty in the autonomous agent's dynamics as a POMDP with uncertain transition and/or observation probabilities. The class of uncertainties we study belong to an interval~\cite{ITOH2007453}. Let $T_u$ denote the set of triplets $(q,a,q')$  corresponding to the uncertain transition probabilities. Similarly, let $O_u$ denote the set of triplets $(q,a,z)$ corresponding to the uncertain observation probabilities. We consider the class of POMDPs with the following interval transition and/or observation probabilities 
\begin{subequations}\label{eq:uncertaintpdf}
		\begin{equation} 
		T(q,a,q') \in [\underline{l}_{q,a,q'},\overline{l}_{q,a,q'}],~ (q,a,q') \in T_u,
		\end{equation}	
		\begin{equation} 
		O(q,a,z) \in [\underline{o}_{q,a,z},\overline{o}_{q,a,z}],~(q,a,z) \in O_u,
		\end{equation}	
		\end{subequations}
where the constants $0\le\underline{l}_{q,a,q'}\le \overline{l}_{q,a,q'}\le1$ for all $(q,a,q') \in T_u$ and $0\le\underline{o}_{q,a,z}\le \overline{o}_{q,a,z}\le1$ for all $(q,a,z) \in O_u$.

In the sequel, we focus on the case of uncertain transition probabilities, but the extension to the case of uncertain observation transition  probabilities is straightforward and follows the same lines.




\subsection{Safety and Optimality}

For typical POMDP problems, we are often interested in assessing both optimal and safe behavior.
In the following, we define the formal notions of optimality and safety we consider here.
  
We define \emph{safety} as the probability of reaching a set of unsafe states $Q_u \subset Q$ being less than a given constant. 
To this end, we use the belief states.
Formally, we are interested in solving the following problem.
\begin{problem}
Given an uncertain POMDP with interval probabilities as described in~\eqref{eq:uncertaintpdf}, a point future in time $t^*$, a set of unsafe states $Q_u$, and a safety requirement constant $\lambda$, check whether
\begin{equation}\label{eq:safety}
g \left(b_{t^*}(q) \right) \le \lambda,~~q \in Q_u,
\end{equation}
where $g:\mathcal{B} \to \mathbb{R}$. In particular, $g$ can be an affine function.
\end{problem}
\smallskip
%

In addition to safety, we are interested in checking whether an \emph{optimality} criterion is satisfied.  
\begin{problem}
Given an uncertain POMDP with interval probabilities as described in~\eqref{eq:uncertaintpdf}, the reward function $R:Q \times A \to \mathbb{R}$, in which $R(q,a)$ denotes the reward of taking action $a$ while being at state $q$, a point future in time $t^*$, and a optimality requirement $\gamma$, check whether
\begin{equation}\label{eq:optimality}
\sum_{s=0}^{t^*} r(b_s,a_s) \le \gamma,
\end{equation}
where $r(b_s,a_s) = \sum_{q \in Q} b_t R(q,a_t)$.
\end{problem}
\smallskip
%


\section{Main Results}\label{sec:main}

Checking whether~\eqref{eq:safety} and~\eqref{eq:optimality} hold by solving the POMDP directly is a PSPACE-hard problem~\cite{ChatterjeeCT16}, not to mention the difficulties arising from uncertain transition probabilities.  In this section, we first demonstrate that POMDPs can be represented as discrete-time switched systems. Then, we borrow a notion from control theory to check the safety and/or optimality requirements of a given POMDP with a guarantee or a certificate. 

\subsection{Treating POMDPs as Switched Systems}

The belief update equation~\eqref{equation:belief update}  is a discrete-time switched system, where the actions $a \in A$ define the switching modes.  Formally, the belief \emph{dynamics}~\eqref{equation:belief update} can be described as
\begin{equation}\label{equation:discretesystem1}
b_t = f_a\left(b_{t-1},z_t\right),
\end{equation}
where $b$ denote the belief vector belonging to the belief unit simplex $\mathcal{B}$ and $b_0 \in \mathcal{B}_0 \subset \mathcal{B}$ representing the set of initial beliefs (prior). In~\eqref{equation:discretesystem}, $a \in A$ denote the actions that can be interpreted as the switching modes, $z \in Z$ are the observations representing inputs, and $t \in \mathbb{Z}_{\ge 1}$ denote the discrete time instances. The vector fields $\{f_{a}\}_{a \in A}$ with $f_a: [0,1]^{|Q|} \to [0,1]^{|Q|} $ are described as the vectors with rows
$$
f_a^{q'}(b,\cdot,z) = \frac{O(q',a,z)\sum_{q\in Q}T(q,a,q')b_{t-1}(q)}{\sum_{q'\in Q}O(q',a,z)\sum_{q\in Q},T(q,a,q')b_{t-1}(q)},
$$
where $f_a^{q'}$ denotes the $q'$th row of $f_a$.  If the transition probabilities are uncertain, i.e., they belong to some given set, the system can be represented as an uncertain discrete-time switched system
\begin{equation}\label{equation:discretesystem}
b_t = f_a\left(b_{t-1},\theta,z_t\right),
\end{equation}
where $\theta \in \Theta$ is a set of uncertain parameters and $\Theta$ represents the uncertain transition probability intervals \eqref{eq:uncertaintpdf}. That is, 
\begin{equation*}
		\theta_{q,a,q'} = T(q,a,q') \in [\underline{l}_{q,a,q'},\overline{l}_{q,a,q'}],~ (q,a,q') \in T_u,
\end{equation*}
and
\begin{equation}\label{eq:uncertainparams}
\Theta =\left \{ \theta \mid \theta_{q,a,q'}  \in  [\underline{l}_{q,a,q'},\overline{l}_{q,a,q'}],~ (q,a,q') \in T_u \right \}.
\end{equation}

In this study, we consider two classes of problems in POMDP verification:
\begin{itemize}
\item [1.] No policy is given: This case corresponds to analyzing~\eqref{equation:discretesystem} under \emph{arbitrary switching} with switching modes given by $a \in A$.
\item [2.] A policy is given: This corresponds to to analyzing~\eqref{equation:discretesystem} under \emph{state-dependent switching}. Indeed, the policy $\pi:\mathcal{B} \to A$ determines regions in the belief space where each mode (action) is active.
\end{itemize}

Both cases of switched systems with~\emph{arbitrary switching} and~\emph{state-dependent switching} are well-known in the systems and controls literature~\cite{liberzon2003switching}. The next example illustrates the proposed  switched system representation  for POMDPs with a given policy.

\begin{example}Consider a POMDP with two states $\{q_1,q_2\}$, two actions $\{a_1,a_2\}$, and $z \in {Z}$. The policy
\begin{equation} \label{eq:example-sds}
\pi=
\left\{\begin{array}{lr}
        a_1, & b \in \mathcal{B}_1,\\
        a_2, & b \in \mathcal{B}_2\\
        \end{array}\right.
\end{equation}
leads to different switching modes based on whether the states belong to the regions $\mathcal{B}_1$ or $\mathcal{B}_2$ (see Figure~\ref{figure1}). That is, the belief update equation~\eqref{equation:discretesystem} is given by 
\begin{equation} \label{eq:example-dynamics}
b_t=
\left\{\begin{array}{lr}
        f_{a_1}\left(b_{t-1},z_t\right), & b \in \mathcal{B}_1,\\
        f_{a_2}\left(b_{t-1},z_t\right), & b \in \mathcal{B}_2.\\
        \end{array}\right.
\end{equation}
Note that the belief space is given by $\mathcal{B}=\mathcal{B}_1 \cup \mathcal{B}_2 = \{ b \mid b(q_1)+b(q_2)=1\}$.
\end{example}
\begin{figure}[tbp] 
\begin{center} 
\includegraphics[width=5cm]{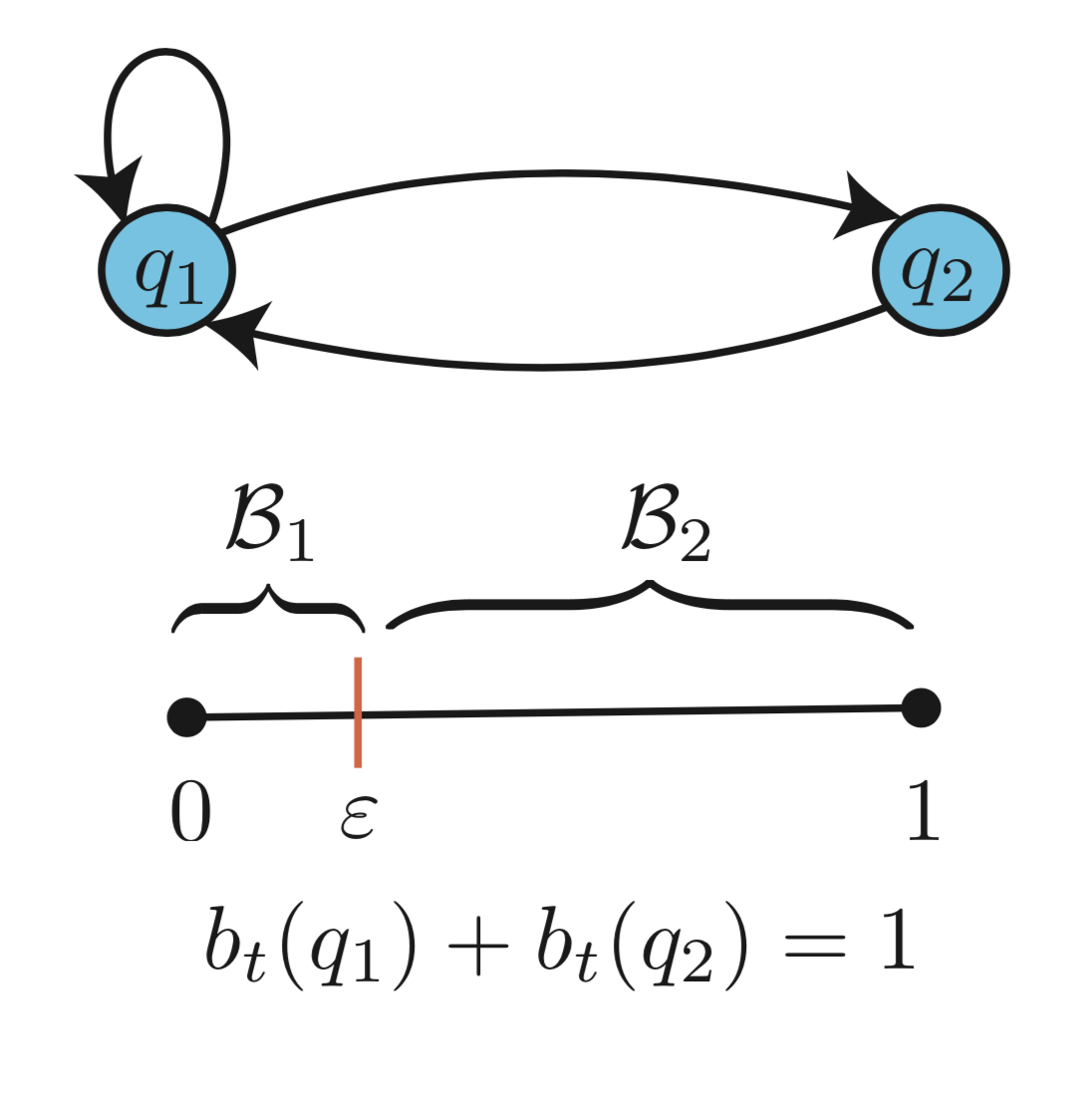}
\caption{An example of a POMDP with two states and the state-dependent switching modes induced by the policy~\eqref{eq:example-sds}. }
\label{figure1}
\end{center}
\end{figure}

\subsection{Verification Using Barrier Certificates}

In the following, we show how we can use barrier certificates to verify properties of the switched systems induced by POMDPS. 

Let us define the following unsafe set
\begin{equation}\label{eq:unsafeset}
\mathcal{B}_u^s = \{ b \in \mathcal{B} \mid g\left( b_{t^*}(q) \right) > \lambda,~~q \in Q_u \},
\end{equation}
which is the complement of~\eqref{eq:safety}.

\begin{thm}\label{theorem-barrier-discrete}
Consider the belief update equation~\eqref{equation:discretesystem} and the uncertain transition probabilities~\eqref{eq:uncertainparams}. Given a set of initial beliefs $\mathcal{B}_0 \subset [0,1]^{|Q|}$, an unsafe set $\mathcal{B}^s_u$ as given in~\eqref{eq:unsafeset} ($\mathcal{B}_0 \cap \mathcal{B}^s_u = \emptyset$), and a constant $t^*$, if there exists a function $B:\mathbb{Z} \times \mathcal{B} \to \mathbb{R}$ called the barrier certificate such that
\begin{equation}\label{equation:barrier-condition1}
B(t^*,b_{t^*})  > 0, \quad \forall b_{t^*} \in \mathcal{B}^s_u,
\end{equation}
\begin{equation}\label{equation:barrier-condition11}
 B(0,b_0) < 0, \quad \forall b_0 \in \mathcal{B}_0,
\end{equation}
and
\begin{multline}\label{equation:barrier-condition2}
B\left(t,f_a(b_{t-1},\theta,z)\right) - B(t-1,b_{t-1}) \le 0, \\ \forall t \in \{1,2,\ldots,t^*\},\\ \forall a \in A,~\forall \theta \in \Theta,~\forall z \in Z, ~\forall b \in \mathcal{B},
\end{multline}
then there exist no solution of the belief update equation~\eqref{equation:discretesystem} such that $b_0 \in \mathcal{B}_0$, and $b_{t^*} \in \mathcal{B}_u$ for all $a \in A$ and all $\theta \in \Theta$.
\end{thm}
\begin{proof}
The proof is carried out by contradiction. Assume at time instance ${t^*}$ there exist a solution to \eqref{equation:discretesystem} such that $b_0 \in \mathcal{B}_0$ and $b_{t^*} \in \mathcal{B}^s_u$. From inequality~\eqref{equation:barrier-condition2}, we have 
$$
B(t,b_t) \le B(t-1,b_{t-1})
$$
for all $t\in \{1,2,\ldots,{t^*}\}$, all actions $a \in A$, and $\theta \in \Theta$. Hence, $B(t,b_t) \le B(0,b_0)$ for all $t \in \{1,2,\ldots,{t^*}\}$. Furthermore, inequality~\eqref{equation:barrier-condition11} implies that 
$$
B(0,b_0) < 0 
$$
for all $b_0 \in \mathcal{B}_0$.  Since the choice of ${t^*}$ can be arbitrary, this is a contradiction because it implies that $B({t^*},b_{t^*}) \le B(0,b_0) < 0$. Therefore, there exist no solution of \eqref{equation:discretesystem} such that $b_0 \in \mathcal{B}_0$ and $b_{t^*} \in \mathcal{B}^s_u$ for any sequence of actions $a \in A$ and uncertain  probabilities belonging to $\Theta$. Therefore, the safety requirement is satisfied.
\end{proof}

The above theorem provides conditions under which the POMDP is \emph{guaranteed} to be safe. The next result brings forward a set of conditions, which verifies whether the optimality criterion~\eqref{eq:optimality} is satisfied. 

\begin{cor}\label{corollary-barrier-discrete:optimality}
Consider the belief update equation~\eqref{equation:discretesystem}, the uncertain  probabilities~\eqref{eq:uncertainparams} and the optimality criterion $\gamma$ as given by~\eqref{eq:optimality}. Let $\tilde{\gamma}:\mathbb{Z}_{\ge 0} \to \mathbb{R}$ satisfying
\begin{equation} \label{eq:contstraint-on-gamma}
\sum_{s=0}^{t^*} \tilde{\gamma}(s) \le \gamma.
\end{equation}
Given a set of initial beliefs $\mathcal{B}_0 \subset \mathcal{B}$, an unsafe set
\begin{equation}\label{eq:parametrized-unsafe-set}
{\mathcal{B}}_u^o = \left \{ (t,b) \mid r(b_t,a_t) > \gamma(t) \right\},
\end{equation}
 and a constant $t^*$, if there exists a function $B:\mathbb{Z} \times \mathcal{B} \to \mathbb{R}$  such that~\eqref{equation:barrier-condition1}-\eqref{equation:barrier-condition2} are satisfied with $\mathcal{B}^o_u$ instead of $\mathcal{B}^s_u$, then for all $b_0 \in \mathcal{B}_0$ the optimality criterion~\eqref{eq:optimality} holds.
\end{cor}
\begin{proof}
The proof is straightforward and an application of Theorem~\ref{theorem-barrier-discrete}. If conditions~\eqref{equation:barrier-condition1}-\eqref{equation:barrier-condition2} are satisfied with $\mathcal{B}^o_u$ instead of $\mathcal{B}^s_u$, based on Theorem~\ref{theorem-barrier-discrete}, we conclude that there exist no solution of the belief update equation~\eqref{equation:discretesystem} such that $b_0 \in \mathcal{B}_0$, and $b_{t^*} \in {\mathcal{B}}_u^o$ for all $a \in A$ and all $\theta \in \Theta$. Therefore, we have
$$
r(b_t,a_t) \le \tilde{\gamma}(t), \quad \forall t \in \{0,1,\ldots,t^*\}.
$$
Summing up both sides of the above equation from $t=0$ to $t=t^*$ yields
$$
\sum_{s=0}^{t^*} r(b_s,a_s) \le \sum_{s=0}^{t^*} \tilde{\gamma}(s).
$$
Then, from \eqref{eq:contstraint-on-gamma}, we conclude that $\sum_{s=0}^{t^*} r(b_s,a_s) \le \gamma$.
\end{proof}

\begin{figure}[tbp] 
\begin{center} 
\includegraphics[width=8.5cm]{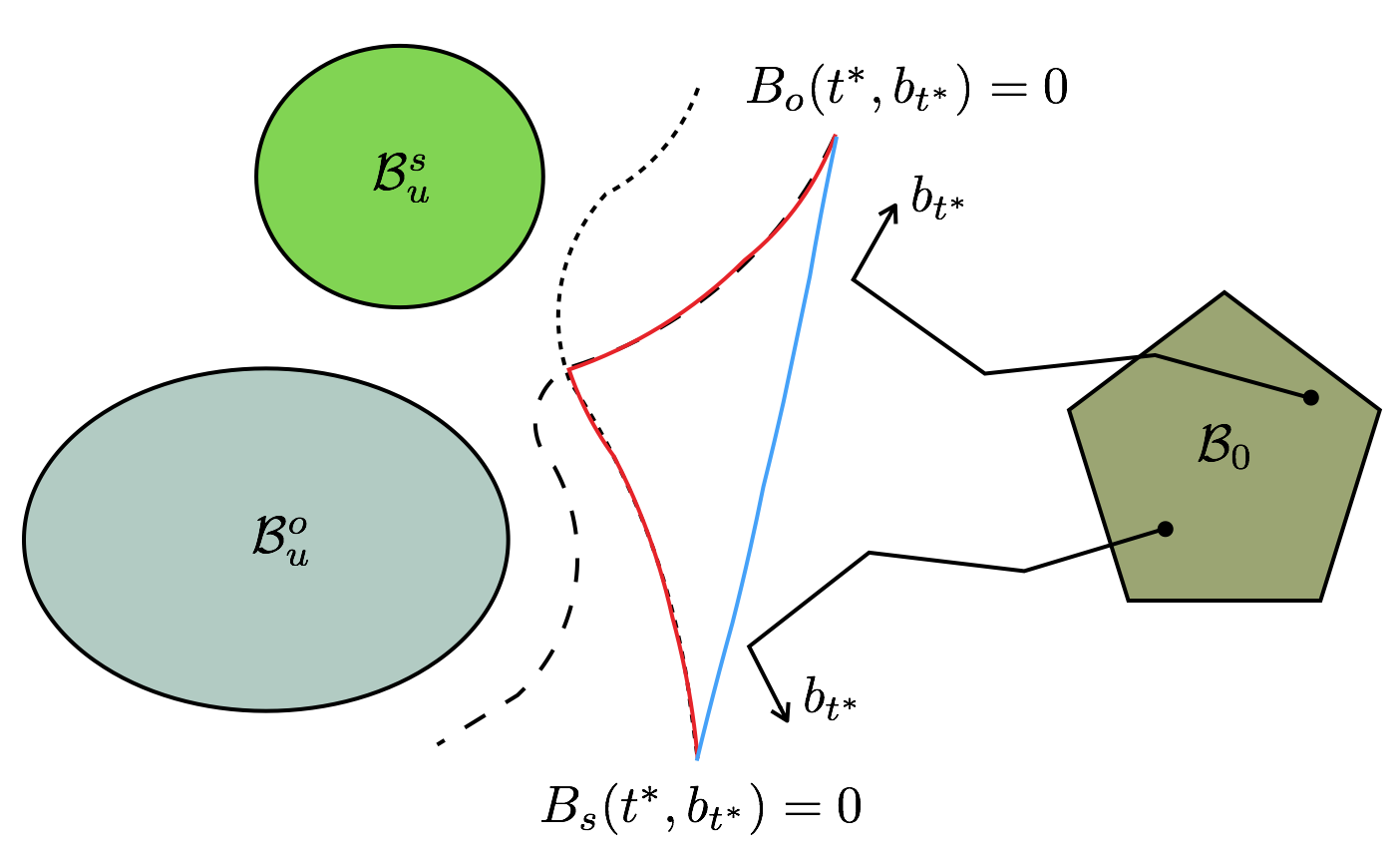}
\caption{Two methods for ensuring both safety and optimality. The zero-level sets of $B_s$ ($B_o$) separate the evolutions of the beliefs starting at $\mathcal{B}_0$ from $\mathcal{B}_u^s$  ($\mathcal{B}_u^o$). The red line illustrates the zero-level sets of the barrier certificate formed by taking the maximum of $B_s$ and $B_0$. The blue line illustrate the zero-level set of the barrier certificate formed by taking the convex hull of $B_s$ and $B_0$.}
\label{figure2}
\end{center}
\end{figure}
The technique used in Corollary~\ref{corollary-barrier-discrete:optimality} is analogous to the one used in~\cite{7171125,AHMADI201733} for bounding (time-averaged) functional outputs of systems described by partial differential equations. The method proposed here, however, can be used for a large class of discrete time systems and the belief update equation is a special case that is of our interest.

In practice, it is often desirable to make sure a design is both optimal and safe. The problem can be described by checking whether the solutions of the belief update switched dynamics~\eqref{equation:discretesystem} enter the following set 
$$
\mathcal{B}_u = \mathcal{B}_u^s \cup \mathcal{B}_u^o.
$$
To this end, we can adopt either of the following approaches (see Figure~\ref{figure2}). Both of these approaches are based on the construction of non-smooth barrier certificates. The first one, proposed in~\cite{7937882}, suggests finding the barrier certificate for $\mathcal{B}_u^s$ and $\mathcal{B}_u^o$ separately or in parallel. The barrier certificate for the set $\mathcal{B}_u$ is then the maximum of the two certificates, i.e., $B = \max \{B_s,B_o\}$, where $B_s$ is the barrier certificate for checking safety and $B_o$ is the barrier certificate for checking optimality. The second method proposed by the authors in~\cite{8264626,2018arXiv180104072A} suggests searching for a barrier certificate composed of the convex hull of the $B_s$ and $B_o$. In this paper, we adopt the latter method.


\subsection{Computational Method based on \\Sum-of-Squares Programming} \label{sec:SOS}

The belief update equation~\eqref{equation:belief update} is a rational function in the belief states $b_t(q)$, $q \in Q_s$
\begin{multline}\label{eq:belief-update-rational}
b_t(q') = \frac{S_a\left( b_{t-1}(q'),\theta,z_{t-1} \right)}{R_a\left( b_{t-1}(q'),\theta,z_{t-1}  \right)} \\
= \frac{O(q',a_{t-1},z_{t})\sum_{q\in Q}T(q,a_{t-1},q')b_{t-1}(q)}{\sum_{q'\in Q}O(q',a_{t-1},z_{t})\sum_{q\in Q}T(q,a_{t-1},q')b_{t-1}(q)}.
\end{multline}
The uncertain probabilities $T(q,a_{t-1},q')$ are parameters that belong to the set~\eqref{eq:uncertainparams}. Moreover, the unsafe  set~\eqref{eq:unsafeset} and the uncertainty set~\eqref{eq:uncertainparams}  are semi-algebraic sets, since they can be described by  polynomial inequalities. We further assume the set of initial beliefs is also given by a semi-algebraic set as follows
\begin{equation}\label{eq:intitial-belief-semialgebraic}
\mathcal{B}_0 = \bigg \{ b_0 \in \mathbb{R}^{|Q_s|} \mid l_i^0(b_0) \le 0,~i = 1,2,\ldots, n_0 \bigg\},
\end{equation}
and $g \in \mathcal{R}[b]$ as in~\eqref{eq:safety}.

At this stage, we are ready to present conditions based on sum-of-squares programs to verify safety of a given uncertain POMDP. 

\begin{cor}\label{cor:SOS-Safety}
Consider the POMDP belief update dynamics~\eqref{eq:belief-update-rational}, the  unsafe set~\eqref{eq:unsafeset}, the set of uncertain probabilities~\eqref{eq:uncertainparams}, the set of initial beliefs \eqref{eq:intitial-belief-semialgebraic}, and a constant $t^*>0$. If there  exist polynomial functions $B \in \mathcal{R}[t,b]$ of degree $d$,   $p^u_q \in {\Sigma}[b]$, $q \in Q_u$,  $p_i^0 \in {\Sigma}[b]$, $i = 1,2,\ldots, n_0$, $p^\theta_{q,a,q'} \in  \Sigma[b,\theta] $, $(q,a,q') \in T_u$, and constants $s_1,s_2>0$ such that
\begin{equation}\label{eq:setssos1}
B\left({t^*},b_{t^*}\right) +  \sum_{q\in Q_u} p^u_q(b_{t^*}) \left( g\left(b_{t^*}(q)\right) - \lambda \right)- s_1 \in \Sigma \left[b_{t^*}\right],
\end{equation}
\begin{equation}\label{eq:setssos2}
-B\left(0,b_0\right) + \sum_{i=1}^{n_0} p_i^0(b_0) l_i^0(b_0) - s_2 \in \Sigma \left[b_0\right],
\end{equation}
and 
\begin{multline}\label{eq:setssos3}
- {R_a\left( b_{t-1} \right)}^d\bigg(B\left(t,\frac{S_a\left( b_{t-1}, \theta,z \right)}{R_a\left( b_{t-1},\theta,z \right)} \right) - B(t-1,b_{t-1}) \\ - \sum_{(q,a,q') \in T_u} p^\theta_{q,a,q'}(\theta,b_{t-1})(\underline{l}_{q,a,q'} - \theta_{q,a,q'})(\overline{l}_{q,a,q'}-\theta_{q,a,q'})\bigg)  \\  \in \Sigma[t,b_{t-1},\theta],  \forall  t \in \{1,2,\ldots,{t^*}\},~z \in Z,~a \in A,
\end{multline}
then there exists no $b_0 \in \mathcal{B}_0$ such that $b_{t^*} \in \mathcal{B}_u$.
\end{cor}
\begin{proof}
 Sum-of-squares conditions~\eqref{eq:setssos1} and~\eqref{eq:setssos2} are a direct consequence of applying Propositions~\ref{chesip} and~\ref{spos} in Appendix~A to verify conditions \eqref{equation:barrier-condition1} and \eqref{equation:barrier-condition11}, respectively. Furthermore, condition~\eqref{equation:barrier-condition2} for system \eqref{eq:belief-update-rational} can be re-written~as  
\begin{multline*}
B\left(t,\frac{S_a\left( b_{t-1},\theta,z \right)}{R_a\left( b_{t-1},\theta,z \right)} \right) - B(t-1,b_{t-1})>0,\\ \forall a \in A,~\forall \theta \in \Theta,~\forall z \in Z.
\end{multline*}
Since $\theta \in \Theta$ is a semi-algebraic set, we use Propositions~\ref{chesip} and~\ref{spos} in Appendix~A to obtain
\begin{multline*}
B\left(t,\frac{S_a\left( b_{t-1}, \theta,z \right)}{R_a\left( b_{t-1},\theta,z \right)} \right) - B(t-1,b_{t-1}) \\ - \sum_{(q,a,q') \in T_u} p^\theta_{q,a,q'}(\theta,b_{t-1})(\underline{l}_{q,a,q'} - \theta_{q,a,q'})(\overline{l}_{q,a,q'}-\theta_{q,a,q'}) \\ \in \Sigma[t,b_{t-1},\theta], \forall a \in A,~\forall z \in Z.
\end{multline*}
for $p^\theta_{q,a,q'} \in  \Sigma[b,\theta] $, $(q,a,q') \in T_u$.
Given that $R_a\left( b_{t-1}(q'),\theta,z \right)$ is a positive polynomial of degree one, we can relax the above inequality into a sum-of-squares condition given by
\begin{multline}
- {R_a\left( b_{t-1},\theta,z  \right)}^d\bigg(B\left(t,\frac{S_a\left( b_{t-1},\theta,z  \right)}{R_a\left( b_{t-1},\theta,z  \right)} \right) - B\left(t-1,b_{t-1} \right) \\
- \sum_{(q,a,q') \in T_u} p^\theta_{q,a,q'}(\theta,b_{t-1})(\underline{l}_{q,a,q'} - \theta_{q,a,q'})(\overline{l}_{q,a,q'}-\theta_{q,a,q'})
 \bigg) \\  \in \Sigma[t,b_{t-1},\theta]. \nonumber
\end{multline}
Hence, if ~\eqref{eq:setssos3} holds, then~\eqref{equation:barrier-condition2}  is satisfied as well. From Theorem~\ref{theorem-barrier-discrete}, we infer that there is no $b_t(q)$ at time $t^*$ such that $b_0(q) \in \mathcal{B}_0$ and  $g\left(b_{t^*}(q) \right)> \lambda$. Equivalently, the safety requirement is satisfied at time $t^*$. That is, $g\left(b_{t^*}(q) \right) \le \lambda$.
\end{proof}

Checking whether optimality holds can also be cast into sum-of-squares programs. To this end, we assume the reward function is a polynomial (or can be approximated by a polynomial\footnote{This assumption is realistic, since the beliefs belong to a bounded set (a unit simplex) and by Stone-Weierstrass theorem any continuous function defined on a bounded domain can be uniformly approximated arbitrary close by a polynomial~\cite{Stone37}.}) in beliefs , i.e., $R \in \mathcal{R}[b]$.

The following Corollary can be derived using similar arguments as the proof of Corollary~\ref{cor:SOS-Safety}.

\begin{cor}
Consider the POMDP belief update dynamics~\eqref{eq:belief-update-rational}, the set of uncertain probabilities~\eqref{eq:uncertainparams},  the set of initial beliefs \eqref{eq:intitial-belief-semialgebraic}, and a constant $t^*>0$. If there  exist polynomial functions $\tilde{\gamma} \in \mathcal{R}[t]$ characterizing the unsafe set~\eqref{eq:parametrized-unsafe-set}, $B \in \mathcal{R}[t,b]$ with degree $d$,   $p^u_q \in {\Sigma}[b]$, $q \in Q_u$,  $p_i^0 \in {\Sigma}[b]$, $i = 1,2,\ldots, n_0$, $p^\theta_{q,a,q'} \in  \Sigma[b,\theta] $, $(q,a,q') \in T_u$, and constants $s_1,s_2>0$ such that \eqref{eq:contstraint-on-gamma}, and \eqref{eq:setssos1}-\eqref{eq:setssos3} are satisfied, then for all $b_0 \in \mathcal{B}_0$, the optimality criterion~\eqref{eq:optimality} holds.
\end{cor}

\section{Numerical Experiment} \label{sec:example}

We demonstrate the applicability of our methods on a variant of the \emph{RockSample} problem~\cite{smith2004heuristic}. 
We model the problem using the input language of the probabilistic model checker PRISM~\cite{KNP11}. 
In a Python toolchain, we employ the model checker Storm~\cite{DBLP:conf/cav/DehnertJK017} to build the explicit state space of the examples.
In order to check the sum-of-squares conditions formulated in Section~\ref{sec:SOS}, we use diagonally-dominant-sum-of-squares (DSOS) relaxations of the sum-of-squares programs implemented through the Systems Polynomial Optimization Toolbox (SPOT)~\cite{spot} (for more details see~\cite{ahmadi2017dsos,8263706}). 

\subsection{Uncertain POMDP Model for Mars Rover Exploration}

A Mars rover explores a terrain, where ``scientifically valuable'' rocks may be hidden. 
The locations of the rocks are known, but it is unknown whether they have the type ``good'' or ``bad''. 
Once the rover moves to the immediate location of a rock, it can sample its type. 
As sampling is expensive, the rover is equipped with a noisy long-range sensor that returns an estimate on the type of the rock. 
The accuracy of the sensor decreases with the distance to the rock.

\begin{figure}
	\centering
	\includegraphics[scale=0.3]{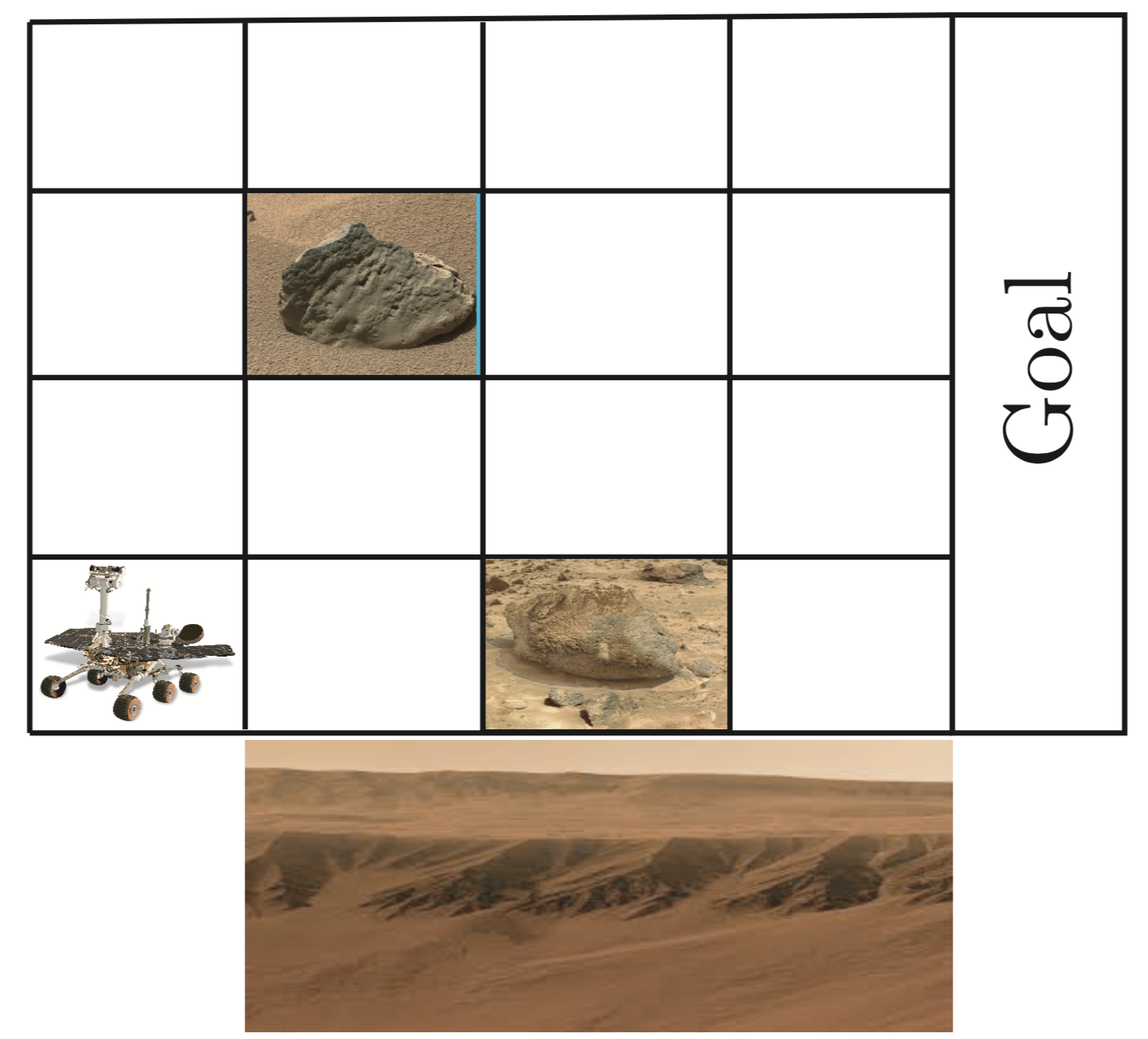}
	\caption{A \emph{RockSample}[$4,2$] instance where the initial position of the rover and the two rock positions in the grid are known. To the right is the goal area, and to the lower side of the grid is a sand dune from which the rover may fall over.}
	\label{fig:marsrover}
\end{figure}

Formally, \emph{RockSample}[$n,k$] describes an instance of the problem with the terrain being a grid of size $n\times n$ and $k$ rocks, which may have one of the types $\mathit{RockType}_i=\{\mathit{good},\mathit{bad}\}$ for $1\leq i\leq k$.
The rover may choose from the actions $\{\mathit{Up}, \mathit{Down}, \mathit{Left}, \mathit{Right}, \mathit{Sample}, \mathit{Check_1},\ldots,\mathit{Check_k}\}$.
When the rover moves off the right edge of the grid, it reaches its \emph{goal} area, where it receives a \emph{reward} of $10$. 
Sampling of a rock yields a reward of $10$ if the rock is $\mathit{good}$, and $-10$ otherwise. 
The potentially negative reward causes an incentive to predict the type of a rock in advance.
Executing action $\mathit{Check_i}$ returns a noisy \emph{observation} whether rock $i$ is $\mathit{good}$ or $\mathit{bad}$.
The probability of a wrong observation decreases with the distance to rock $i$.
The underlying model is a POMDP, where the positions of the rover and the rock are observable, while the type of the rocks is not observable unless a rock has been sampled.
The \emph{belief} describe the probability of the correct rock types. 
To maximize the (expected) reward, the rover aims to correctly estimate the types in order to not sample a $\mathit{bad}$ rock.

To account for the full potential of our method, we augmented the original \emph{RockSample} problem as described above by (1) uncertainty and (2) safety considerations.
First, concrete probabilities for wrong observations using the long-range sensor seem unrealistic when one considers that they may be the result of simulations and statistical inference of a probabilistic sensor model.
Therefore, we introduce \emph{interval} uncertainties. 
For instance, if from a certain distance there is a probability of $0.5$ for a faulty observation, we may assume that this probability lies within the interval $[0.5,0.6]$ to account for even worse accuracy of the sensor.
A policy that maximizes the reward for the rover should then be \emph{robust} against the uncertainties.

Secondly, we assume the rover has a certain probability to fall off a sand dune located at the bottom of the terrain.
Safety considerations imply that the probability of falling off the dune should be less than $10\%$.
Such scenarios are commonly referred to as \emph{slippery grid worlds}.
The problem for a $4\times 4$ grid and $2$ rocks is depicted in Figure~\ref{fig:marsrover}.

\subsection{Numerical Results}

In our model, we assume the rock at the bottom of the grid is $\mathit{bad}$, and the other  is $\mathit{good}$. 

\subsubsection{Case I}  The first scenario we consider pertains to a policy that checks the type of the rocks from the initial state, then moves to the goal region after having sufficient confidence about the types of the rock according to a nominal observation probability. In order to model, the distance from the the rock positions and limited sensor accuracy, we assume the probability of having a correct observation belongs the interval $[0.1,0.2]$ and we set the nominal observation probability to $0.2$. Given the policy designed for the nominal observation probability, our goal here is to find a lower bound on $\lambda$ satisfying 
$
P\left(\mathit{rock1\_good} \cap \mathit{rock2\_bad}\right) = b_{t^*}(rock1\_good) b_{t^*}(\mathit{rock2\_bad}) \ge  \lambda,
$
which lower-bounds the probability of identifying the rocks correctly. At the same time, we want to make sure that the rover does not move to the three slippery states at the bottom of the grid. We embed this safety constraint as $P_{slip}=b_{t^*}(\mathit{slipping\_state}) \le 0.1$. To this end, we construct two barrier certificate of fixed degree  using Corollary~\ref{cor:SOS-Safety} and perform a line search on the values of $\lambda$. 

Table \ref{table1} demonstrates the results and shows that increasing the degree of the barrier certificates improves the accuracy of the lower bound. Experiments using PRISM and Storm also corroborate the consistency of these results. In the worst case of the uncertain observation probabilities, the nominal policy achieves the values of $P\left(\mathit{rock1\_good} \cap \mathit{rock2\_bad}\right) = 0.61$ at $t^* = 10$ and $P\left(\mathit{rock1\_good} \cap \mathit{rock2\_bad}\right) = 0.84$ at $t^* = 20$. Figures (4.a) and (4.b) show two snapshots of how the Mars rover moves given this policy: the Mars rover stops at the initial position and use sensors to collect information about the rocks and then moves to the goal region.

\begin{table}[!tp]
\begin{center}
\begin{tabular}{c||c||c||c}
$deg(B)$  & 1 & 2 & 3 \\
\hline
$\lambda$ & 0.12 & 0.29  & 0.57
\end{tabular}~~for $t^* =10$.
\begin{tabular}{c||c||c||c}
$deg(B)$  & 1 & 2 & 3 \\
\hline
$\lambda$ & 0.31 & 0.44  & 0.73
\end{tabular}~~for $t^* =20$.
\caption{Numerical results on the lower-bounds on $\lambda$ for two different horizons in Case I.}\label{table1}
\end{center}
\end{table}%


\subsubsection{Case II}   We set  the probability of having a correct observation  to belong to the interval $[0.32, 0.42]$ and the nominal probability is set to $0.42$. The Mars rover is given a policy such that it  first moves closer to the rocks, then checks the type of the rocks using the sensor, and moves to the goal after identifying the rock types. Figures (4.c) and (4.d) show two snapshots of the trajectory of the Mars rover over the grid using the nominal policy. In this case, we are interested to find lower bounds on $\lambda_1$ and $\lambda_2$  satisfying 
$
 b_{10}(rock1\_good)  \ge  \lambda_1,
$
$
 b_{20}(rock2\_bad)  \ge  \lambda_2,
$
which corresponds to the belief in identifying each individual rock accurately. 

Table \ref{table2} presents the results and demonstrates that increasing the degree of the barrier certificates enhances the  lower bounds. These results tally with experiments in PRISM and Storm, which show that in the worst case, we have $ b_{10}(rock1\_good) = 0.92$ and $b_{20}(rock2\_bad)= 0.94$. 

\begin{table}[!tp]
\begin{center}
\begin{tabular}{c||c||c||c}
$deg(B)$  & 1 & 2 & 3 \\
\hline
$\lambda_1$ & 0.37 & 0.65  & 0.84
\end{tabular}
\begin{tabular}{c||c||c||c}
$deg(B)$  & 1 & 2 & 3 \\
\hline
$\lambda_2$ & 0.32 & 0.73  & 0.89
\end{tabular}
\caption{Numerical results on the lower-bounds on $\lambda_1$ and $\lambda_2$ in Case II.}\label{table2}
\end{center}
\end{table}%


\begin{figure}
\begin{centering}
\begin{tabular}{cc}
  \includegraphics[width=35mm]{marsrocksample-01} &   \includegraphics[width=35mm]{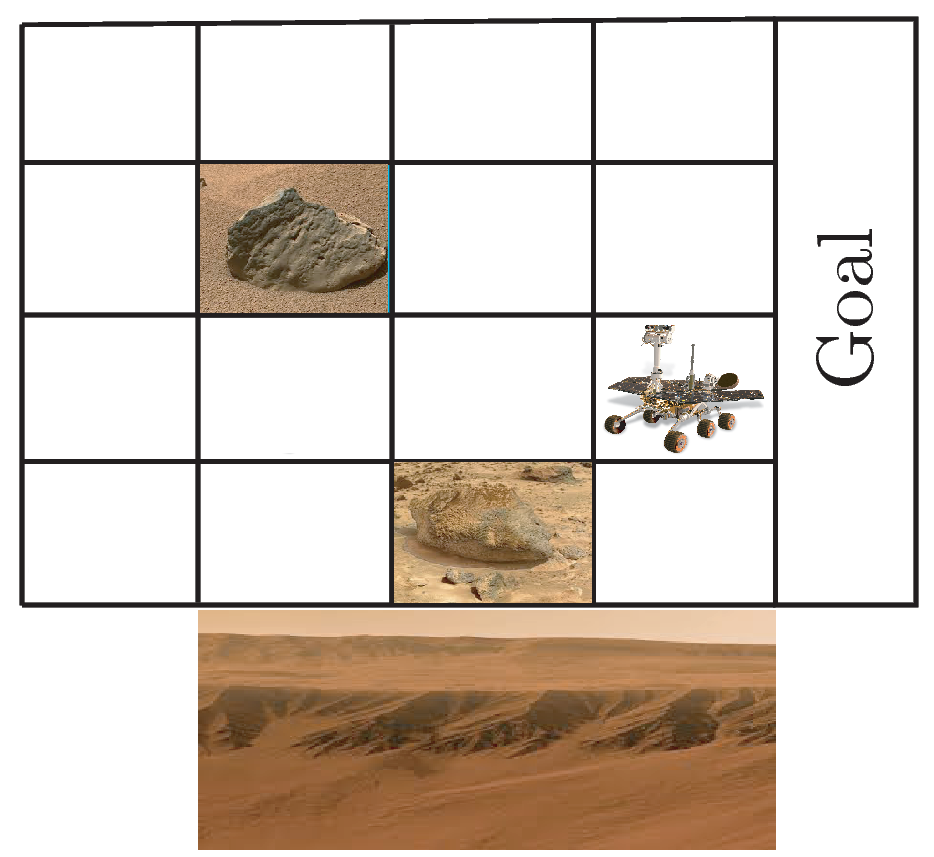}   \\
(a) $t=10$ & (b) $t=45$ \\[6pt]
 \includegraphics[width=35mm]{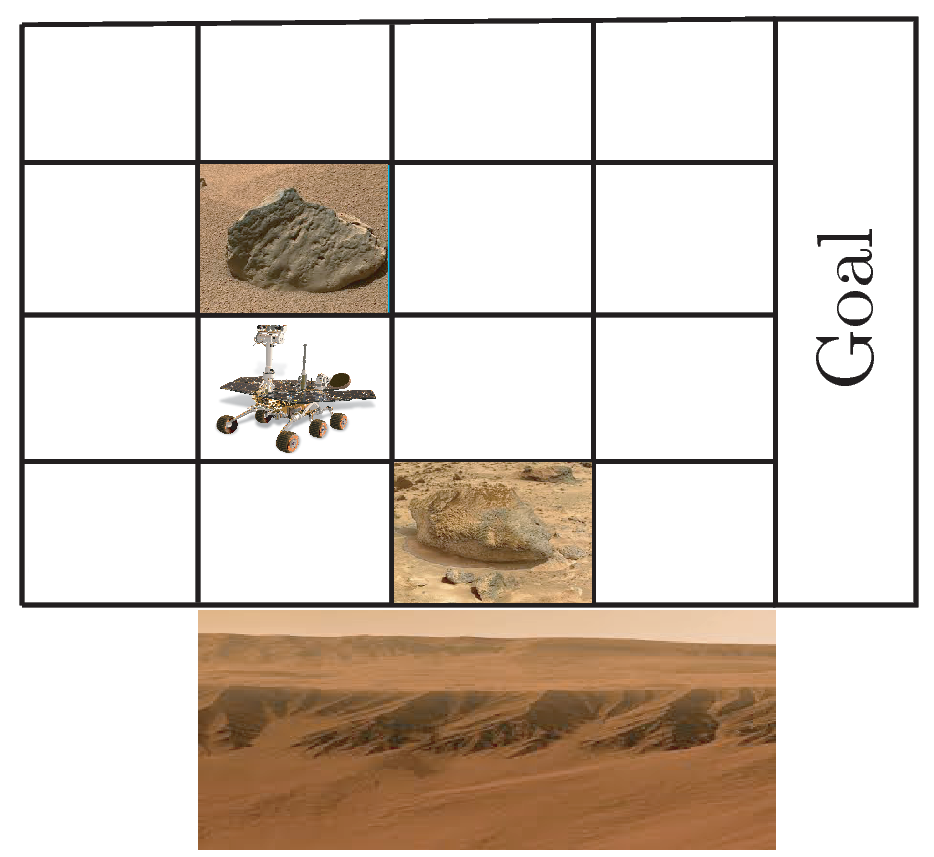} &   \includegraphics[width=35mm]{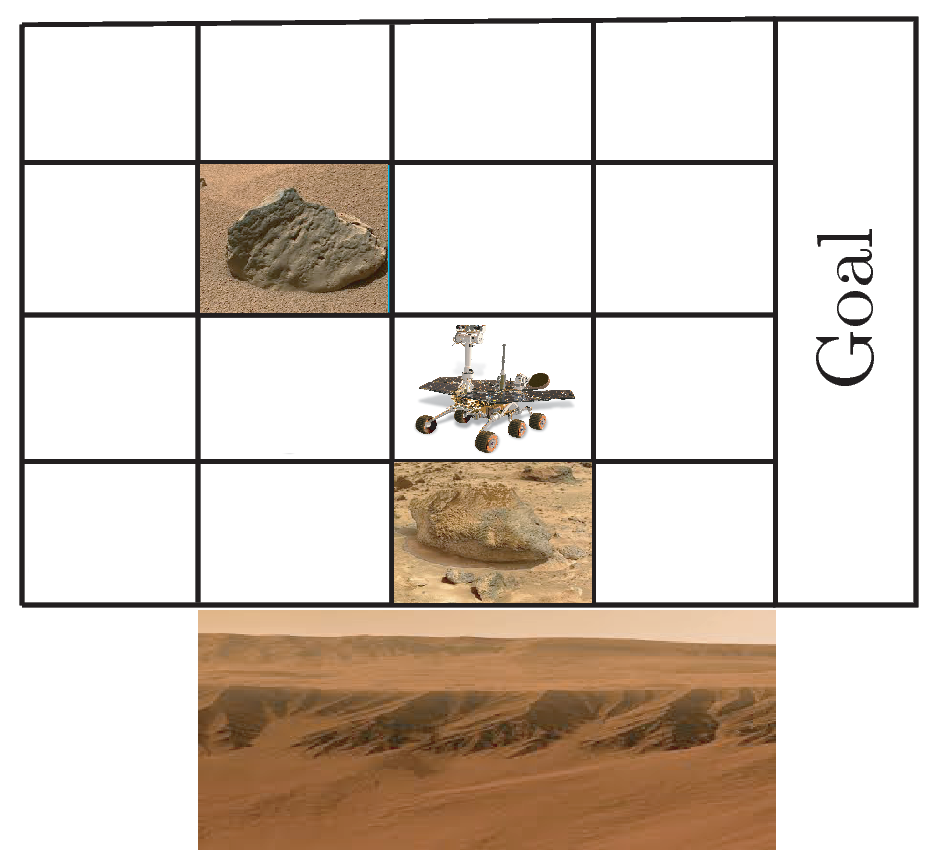} \\
(c) $t=10$ & (d) $t=20$ 
\end{tabular}
\end{centering}
\caption{Top: Positions of the Mars rover at certain time steps $t$ with the first policy. Bottom: Positions of the Mars rover at certain time steps $t$ with the second policy.}
\label{fig:policies}
\end{figure}

\section{CONCLUSIONS AND FUTURE WORK} \label{sec:conclusions}

We proposed an approach for verifying the safety and/or optimality properties of POMDPs with uncertain transition/observation probabilities. The transition and/or observation uncertainties we considered belonged to fixed intervals. We cast the POMDP analysis problem into a switched system analysis problem  and we brought forward a method based on barrier certificates. We showed that  we can verify the satisfaction of optimality or safety requirements by computing a  barrier certificate using sum-of-squares programming. We illustrated the applicability of our method on a Mars rover exploration example.

In this work, we considered the worst case analysis with the uncertain transition and/or observation probabilities. However, this analysis may be too conservative for problems where certain information about the transition/observation probabilities in terms of a probability density function is known.  In this regard, the application of the scenario approach seems relevant~\cite{CAMPI2009149}. Furthermore, the proposed method based on barrier certificates for verification of the POMDPs can also be used to synthesize policies ensuring both safety and optimality.


%


\bibliography{references}
\bibliographystyle{IEEEtran}


\appendix

\subsection{Sum-of-Squares Polynomials} \label{app:SOS}
A polynomial $p(x)$ is a sum-of-squares polynomial if $\exists p_i(x) \in \mathcal{R}[x]$, $i \in \{1, \ldots, n_d\}$ such that $p(x) = \sum_i p_i^2(x)$. Hence $p(x)$ is clearly non-negative. A set of polynomials $p_i$ is called \emph{SOS decomposition} of $p(x)$. The converse does not hold in general, that is, there exist non-negative polynomials which do not have an SOS decomposition~\cite{Par00}.  The computation of SOS decompositions, can be cast as an SDP (see~\cite{Par00,choi1995sums,CTVG99}). The Theorem below proves that, in sets satisfying a property stronger than compactness, any positive polynomial can be expressed as a combination of sum-of-squares polynomials and polynomials describing the set.  

For a set of polynomials $\bar{g} = \{g_1(x), \ldots, g_m(x)\}$, $m \in \nt$, the \emph{quadratic module} generated by $m$ is 
\begin{equation}
M(\bar{g}):= \left\lbrace \sigma_0 +\sum_{i = 1}^{m} \sigma_i g_i | \sigma_i \in \Sigma[x]\right\rbrace.
\end{equation}
A quadratic module $M\in \mathcal{R}[x]$ is said \emph{archimedean} if $\exists N \in \nt $ such that $N - |x|^2 \in M.$ An archimedian set is always compact~\cite{NS08}. At this point, we recall the following result~\cite[Theorem 2.14]{Las09}.

\begin{thm}[Putinar Positivstellensatz]
\label{thm:Psatz}
Suppose the quadratic module $M(\bar{g})$ is archimedian. Then for every $f \in \mathcal{R}[x]$, $$f>0~\forall~x\in \{x | g_1(x)\geq 0, \ldots, g_m(x)\geq 0 \} \Rightarrow f \in (\bar{g}).$$
\end{thm}

The subsequent proposition formalizes the problem of constrained positivity of polynomials which is a direct result of applying Positivstellensatz.
\begin{prop}[\cite{chesi2010lmi}] \label{chesip}
Let $\{a_i\}_{i=1}^k$ and $\{b_i\}_{i=1}^l$ belong to $\mathcal{P}$, then
\begin{eqnarray}
p(x) \ge 0 \quad &\forall x \in \mathbb{R}^n: a_i(x)=0, \, \forall i=1,2,...,k & \nonumber \\
& \text{and} \quad b_j(x) \ge 0, \, \forall j=1,2,...,l&
\end{eqnarray}
is satisfied, if the following holds
\begin{eqnarray} \label{chesieq}
&\exists r_1,r_2,\ldots,r_k \in \mathcal{R}[x] \quad \text{and} \quad \exists s_0,s_1,\ldots,s_l \in \Sigma[x] & \nonumber \\
&p = \sum_{i=1}^k r_i a_i +\sum_{i=1}^l s_i b_i +s_0&
\end{eqnarray}
\end{prop}
\begin{prop} \label{spos}
The multivariable polynomial $p(x)$ is strictly positive ($p(x)>0 \quad \forall x \in \mathbb{R}^n$), if there exists a $\lambda > 0$ such that
\begin{equation}
\big( p(x) - \lambda \big) \in \Sigma[x].
\end{equation}
\end{prop}


\end{document}